\documentclass[%
 reprint,
nobibnotes,
 amsmath,amssymb,
 aps,
]{revtex4-1}

\usepackage{booktabs} 
\usepackage{xcolor}
\usepackage{braket}
\usepackage{bbm}
\usepackage{graphicx}
\usepackage[raggedright,bf]{subfigure}
\usepackage{enumerate}
\usepackage{tikz}
\usepackage{listings}
\usepackage{amsthm}

\newtheorem{theorem}{Theorem}[section]
\newtheorem{lemma}[theorem]{Lemma}
\theoremstyle{definition}
\newtheorem{definition}{Definition}[theorem]

\usetikzlibrary{backgrounds,shadows.blur,fit,decorations.pathreplacing,shapes}

\newcommand{\nix}[1]{{}}

\graphicspath{{./data/},{../data}}

\begin{document}
\title{Managing Approximation Errors in Quantum Programs} 

\author{Thomas H\"aner}
\email{haenert@phys.ethz.ch}
\affiliation{Microsoft Quantum, Microsoft, Redmond, WA 98052, USA}
\affiliation{Institute for Theoretical Physics, ETH Zurich, 8093 Zurich, Switzerland}

\author{Martin Roetteler}
\email{martinro@microsoft.com}
\affiliation{Microsoft Quantum, Microsoft, Redmond, WA 98052, USA}

\author{Krysta M.~Svore}
\email{ksvore@microsoft.com}
\affiliation{Microsoft Quantum, Microsoft, Redmond, WA 98052, USA}

\begin{abstract}
We address the problem of distributing approximation errors in large-scale quantum programs. It has been known for some time that when compiling quantum algorithms for a fault-tolerant architecture, some operations must be approximated as they cannot be implemented with arbitrary accuracy by the underlying gate set. This leads to approximation errors which often can be grouped along subroutines that the given quantum algorithm is composed of. Typically, choices can be made as to how to distribute approximation errors so that the overall error is kept beneath a user- or application-defined threshold. These choices impact the resource footprint of the fault-tolerant implementation. We develop an automatic approximation error management module to tackle the resulting optimization problems. The module is based on annealing and can be integrated into any quantum software framework. Using the benchmark of simulating an Ising model with transverse field, we provide numerical results to quantify the benefits and trade-offs involved in our approach.
\end{abstract}

\maketitle

\section{Introduction}
Quantum computing promises to solve certain problems much faster than classical devices. These problems include factoring large numbers, which is possible using Shor's algorithm~\cite{shor94}, unstructured search problems~\cite{grover1996fast}, and simulation of quantum mechanical systems. The latter area is where the idea of quantum computing originally emerged~\cite{Feynman1982}.
While only few quantum algorithms are known that offer super-polynomial speedups over their classical counterparts, many polynomial speedups have been found over the recent years. However, most of these algorithmic advances take place at a high level of abstraction and often do not take into account constants that may not be explicit in asymptotic analyses. Therefore, while desirable, actual practical resource estimates of complex, large-scale quantum algorithms remain scarce, although they could be used in various ways: to determine the first practical applications of quantum computers, to identify and remedy bottlenecks in existing quantum algorithms, to determine cross-over points (i.e., the smallest problem size for which it is beneficial to use quantum computers), and for hardware/software co-design in general. 

In order to estimate the required resources of a given quantum algorithm, the high-level representation of the algorithm must be translated to a universal low-level set of operations that can be realized in practice. One of the standard gate sets that is often considered is the so-called Clifford+$T$ gate set, which can be generated from a few single-qubit operations and the CNOT (controlled-NOT) gate. In particular, arbitrary single-qubit rotations must be translated to this discrete gate set employing a \textit{rotation synthesis} algorithm, where the resulting gate sequences get longer as the desired accuracy is increased.
Therefore, besides the problem of compiling abstract high-level functions to the native gate set, the resulting approximation errors need to be managed in a way that ensures that the resulting code performs the desired overall operation within a certain (user-specified) tolerance. We address this problem by introducing a method capable of handling these errors automatically.

\vspace{10pt}\noindent\textbf{The need for approximation.}\label{sec:origin}
While it is not possible to perform error correction over a continuous set of quantum operations (gates), this can be achieved over a discrete gate set such as the aforementioned Clifford+$T$ gate set. As a consequence, certain operations must be approximated using gates from this discrete set. An example is the operation which achieves a rotation around the z-axis,
\[
	\text{Rz}_\theta=\left(\begin{matrix}
	e^{-i\theta/2} & 0\\0 & e^{i\theta/2}
	\end{matrix}\right).
\]
To implement such a gate over Clifford+$T$, synthesis algorithms such as the ones in Refs.~\cite{ross2014optimal,KMM1} can be used. Given the angle $\theta$ of this gate, such a rotation synthesis algorithm will produce a sequence of $\mathcal O(\log{\varepsilon_R^{-1}})$ Clifford+$T$ gates which approximate Rz$_\theta$ up to a given tolerance $\varepsilon_R$. We measure approximation error $\varepsilon_R$ with respect to distance in operator norm. As we focus on unitary channels this is equivalent to diamond distance, i.e., approximation errors can be composed safely. 

In most error correction protocols, the $T$-gate is the most expensive operation to realize, as it cannot be executed natively but requires a distillation protocol to distill many noisy magic states into one good state, which can then be used to apply the gate. As a consequence, it is crucial to reduce the number of these $T$-gates as much as possible in order to allow executing a certain quantum computation.

\tikzset{
	linestyle/.style={
		draw,
		thin
	},
	rectitem/.style={
		rectangle,
		minimum height=2.5em,
		minimum width=4em,
		linestyle,
		fill=white,
		blur shadow={shadow blur steps=15, shadow xshift=2pt, shadow yshift=-2pt, shadow scale=.9, shadow blur radius=1.5ex, shadow opacity=60},
	},
}
\newcommand{\txtsize}[0]{\Large}
\newcommand{\txtsizetwo}[0]{\large}
\begin{figure*}[t]
\resizebox{.8\linewidth}{!}{
\begin{tikzpicture}
\node (zero) at (0,-.5) {\txtsize $\varepsilon$};
\node[rectitem,minimum height=3em,minimum width=7em] (QPE) at (0,-2)  {\txtsize QPE($U$)};
\draw[->,linestyle] (zero)--(QPE);

\node[rectitem] (U1) at (-4,-4) {\txtsize $^cU$};
\node at (-1.9,-4.1) {\txtsize $\cdots$};
\node[rectitem] (U3) at (0,-4) {\txtsize $^cU$};
\node[rectitem] (QFT) at (4,-4) {\txtsize $QFT^\dagger$};

\draw[->,linestyle] (QPE)--(U1);
\draw[->,linestyle] (QPE)--(U3);
\draw[->,linestyle] (QPE)--(QFT);

\node at (-2.5,-3) {\txtsizetwo $\varepsilon_1$};
\node at (.4,-3) {\txtsizetwo $\varepsilon_{k-1}$};
\node at (2.5,-3) {\txtsizetwo $\varepsilon_{k}$};

\node[rectitem] (R1) at (-5.3,-6) {\txtsize $R_1$};
\node[rectitem] (R2) at (-2.8,-6) {\txtsize $R_n$};
\node at (-4,-6) {\txtsize $\cdots$};
\draw[->,linestyle] (U1)--(R1);
\draw[->,linestyle] (U1)--(R2);

\node[rectitem] (R21) at (-1.2,-6) {\txtsize $R_1$};
\node[rectitem] (R22) at (1.2,-6) {\txtsize $R_n$};
\node at (-0,-6) {\txtsize $\cdots$};
\draw[->,linestyle] (U3)--(R21);
\draw[->,linestyle] (U3)--(R22);

\node[rectitem] (CR1) at (2.8,-6) {\txtsize $CR_1$};
\node[rectitem] (CR2) at (5.2,-6) {\txtsize $CR_m$};
\node at (4,-6) {\txtsize $\cdots$};
\draw[->,linestyle] (QFT)--(CR1);
\draw[->,linestyle] (QFT)--(CR2);

\node at (-5.15,-5) {\txtsizetwo $\varepsilon_{k+1}$};
\node at (-2.9,-5) {\txtsizetwo $\varepsilon_{k+n}$};
\node at (-1.2,-5) {\txtsizetwo $\varepsilon_{k+n+1}$};
\node at (1.15,-5) {\txtsizetwo $\varepsilon_{k+2n}$};
\node at (2.7,-5) {\txtsizetwo $\varepsilon_{k+2n+1}$};
\node at (5.3,-5) {\txtsizetwo $\varepsilon_{k+2n+m}$};
\end{tikzpicture}
}
\caption{Abstract depiction of the compilation process for a quantum phase estimation (QPE) applied to a given unitary $U$. The parameters $\varepsilon_i$ which get introduced during the compilation must be chosen such that the overall target accuracy $\varepsilon$ is achieved while reducing the resulting cost as much as possible.}
\label{fig:compilation}
\end{figure*}
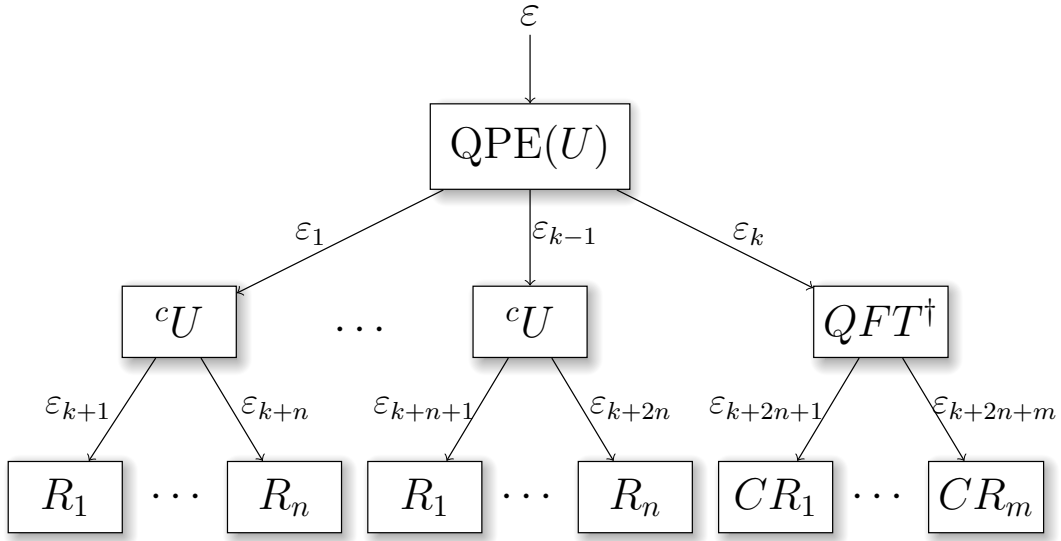

\vspace{10pt}\noindent\textbf{Compilation of quantum programs.}
The job of a quantum program compiler is to translate a high-level description of a given quantum program to hardware-specific machine-level instructions. As in classical computing, such compilation frameworks can be implemented in a hardware-agnostic fashion by introducing backend-independent intermediate representations of the quantum code~\cite{haener2018software}.

During the compilation process, it is crucial to optimize as much as possible in order to reduce the overall depth of the resulting circuit to keep the overhead of the required quantum error correction schemes manageable. Optimizations include quantum versions of constant-folding (such as merging consecutive rotation gates, or even additions by constants) and recognition of compute/action/uncompute sections to reduce the number of controlled gates~\cite{haener2018software}. To allow such optimizations, it is important to introduce multiple layers of abstractions instead of compiling directly down to low-level machine instructions~\cite{haener2018software,Steiger2018projectqopensource}, which would make it impossible to recognize, e.g., two consecutive additions by constants. Even canceling a gate followed by its inverse becomes computationally hard, or even impossible once continuous gates have been approximated.

To translate an intermediate representation to the next lower level of abstraction, a set of decomposition rules is used, some of which introduce additional errors which can be made arbitrarily small at the cost of an increasing circuit size or depth, which in turn implies a larger overhead when applying quantum error correction. Therefore, it is important to choose these error tolerances such that the computation succeeds with high probability given the available resources (number and quality of qubits). See Fig.~\ref{fig:compilation} for an abstract depiction of the compilation process of a quantum phase estimation on a given unitary $U$. At each level of abstraction, the compiler introduces additional accuracy parameters (in the figure denoted by $\varepsilon_i$) which must be chosen such that
\begin{enumerate}[1)]
	\item the overall error lies within the specifications of the algorithm and
	\item the implementation cost is as low as possible while 1) is satisfied.
\end{enumerate}
As mentioned above, it is important to measure approximation errors in a way that is composable to avoid potential issues in underreporting actual approximation errors~\cite{MPGC13,GSVB13} which could be devastating when composing complex quantum algorithms. This leads to diamond distance as the preferred way to measure closeness to the target operation as it composes. As unraveling a complex quantum algorithm eventually leads to primitive gates that are unitary---such as the mentioned $\text{Rz}_\theta$ rotations which are implemented on subsystem of a constant number of qubits---bounding the approximation error in operator norm implies error in diamond norm, i.e., estimates of approximation error can be composed. 

\section{Error-propagation in quantum circuits}\label{sec:error}
The time-evolution of a closed quantum system can be described by a unitary operator. As a consequence, each time-step of our quantum computer can be described by a unitary matrix of dimension $2^n\times 2^n$ (excluding measurement), where $n$ denotes the number of quantum bits (qubits). When decomposing such a quantum operation $U$ into a sequence of lower-level operations $U_M\cdots U_1$, the resulting total error can be estimated from the individual errors $\varepsilon$ of the lower-level gates as follows:

\begin{lemma}\label{lm:unitaryerror}Given a unitary decomposition of $\,U$ such that $U=U_M\cdot U_{M-1}\cdots U_1$ and unitaries $V_i$ which approximate the unitary operators $U_i$ such that $\|V_i-U_i\|<\varepsilon_i\;\forall i$, the total error can be bounded as follows:
\[
	\|U-V_M\cdots V_1\|\leq\sum_{i=1}^M\varepsilon_i.
\]
\end{lemma}
This lemma can be shown straightforwardly from the `hybrid argument'   \cite{BV97} based on triangle inequality and submultiplicativity of $\|\cdot\|$ with $\|U\|\leq 1$. 

\nix{
\begin{proof}
By induction using the triangle inequality and submultiplicativity of $\|\cdot\|$ with $\|U\|\leq 1$. The base case $M=2$ can be proven as follows:
\begin{align*}
\|U_2U_1 - V_2V_1\| &= \|U_2U_1 - U_2V_1 + U_2V_1 - V_2V_1\|\\
&\leq \|U_2(U_1-V_1)\| + \|(U_2-V_2)V_1\|\\
&\leq \varepsilon_1 + \varepsilon_2.
\end{align*}
The induction step $P(M-1)\rightarrow P(M)$ can be shown in a similar fashion
\begin{align*}
\|U_M\cdots U_1 - V_M\cdots V_1\| &= \|U_M\cdots U_1 - U_M\cdots U_2V_1 +\\&\phantom{= \|} U_M\cdots U_2V_1 - V_M\cdots V_1\|\\
&\leq \|U_M\cdots U_2(U_1-V_1)\|+\\&\phantom{=} \|(U_M\cdots U_2 - V_M\cdots V_2)V_1\|\\
&\leq \varepsilon_1 + \sum_{i=2}^M\varepsilon_i=\sum_{i=1}^M\varepsilon_i.
\end{align*}
\end{proof}
}

Note that using only this Lemma in the compilation process to automatically optimize the individual $\varepsilon_i$ would make the resulting optimization problem infeasibly large. What is even worse is that the number of parameters to optimize would vary throughout the optimization process since the number of lower-level gates changes when implementing a higher-level operation at a different accuracy, which in turn changes the number of distinct $\varepsilon_i$. To address these two issues, we introduce Theorem~\ref{thm:decomp} which generalizes Lemma~\ref{lm:unitaryerror}. First, however, we require a few definitions.

\begin{definition}
Let $V_{M(\varepsilon)}\cdots V_1$ be an approximate decomposition of the target unitary $U$ such that $\|V_{M(\varepsilon)}\cdots V_1\|\leq\varepsilon$. A set of subroutine sets $\mathcal S(U, \varepsilon)=\{S_1,...,S_K\}$ is a \textit{partitioning of subroutines of $\,U$} if $\forall i\exists !k : V_i\in S_k$ and we denote by $S(V)$ the function which returns the subroutine set $S$ such that $V\in S$.
\end{definition}

Such a partitioning will be used to assign to each $V_i$ the accuracy $\varepsilon_{S(V_i)}=\varepsilon_{S_k}$ with which all $V_i\in S_k$ are implemented. In order to decompose the cost of $\,U$, however, we also need the notion of a \textit{cost-respecting partitioning of subroutines of $\,U$} and the costs of its subsets:

\begin{definition}
Let $\mathcal S(U, \varepsilon)=\{S_1,...,S_K\}$ be a set of subroutine sets. $\mathcal S(U, \varepsilon)$ is a \textit{cost-respecting partitioning of subroutines of $\,U$} w.r.t. a given cost measure $C(U, \varepsilon)$ if $\forall \varepsilon,i,j,k: ( V_i\in S_k\land V_j\in S_k\Rightarrow C(V_i,\varepsilon)=C(V_j,\varepsilon))$. The cost of a subroutine set $S$ is then well-defined and given by $C(S, \varepsilon):=C(V, \varepsilon)$ for any $V\in S$.
\end{definition}

With these definitions in place, we are now ready to generalize Lemma~\ref{lm:unitaryerror}.
\begin{theorem}\label{thm:decomp}
Let $\mathcal S(U, \varepsilon)=\{S_1,...,S_K\}$ be a cost-respecting partitioning of subroutines for a given decomposition of $\,U$ w.r.t. the cost measure $C(U, \varepsilon)$ denoting the number of elementary gates required to implement $\,U$. Then the cost of $\,U$ can be expressed in terms of the costs of all subroutine sets $S\in\mathcal S(U, \varepsilon_U)$ as follows
\begin{align*}
	C(U, \varepsilon)&=\sum_{S\in \mathcal S(U, \varepsilon_U)} C(S,\varepsilon_{S})f_S(\varepsilon_U)\\
	&\text{with } \sum_{S\in \mathcal S(U, \varepsilon_U)} \varepsilon_S f_S(\varepsilon_U)\leq\varepsilon-\varepsilon_U,
\end{align*}
where $f_S(\varepsilon_U)$ gives the number of subroutines in the decomposition of $\,U$ that are in $S$, given that the decomposition of $\,U$ would introduce error $\varepsilon_U$ if all subroutines were to be implemented exactly and $\varepsilon_S$ denotes the error in implementing subroutines that are in~$S$.
\end{theorem}
\begin{proof}
It is easy to see that the cost $C(U,\varepsilon)$ can be decomposed into a sum of the costs of all subroutines $V_i$. Furthermore, since $\varepsilon_{V}=\varepsilon_{S}$ $\forall V\in S$,
\begin{align*}
	C(U,\varepsilon)&=\sum_i C(V_i, \varepsilon_{V_i})\\&=\sum_i C(V_i, \varepsilon_{S(V_i)})\\&=\sum_{S\in\mathcal S}|\{i:V_i\in S\}|C(S, \varepsilon_{S})
\end{align*}
and $f_S(\varepsilon_U):=|\{i:V_i\in S\}|$ $\forall S\in\mathcal S(U,\varepsilon_U)$.

\noindent
To prove that the overall error remains bounded by $\varepsilon$, let $\tilde U$ denote the unitary which is obtained by applying the decomposition rule for $U$ with accuracy $\varepsilon_U$, i.e., $\|U-\tilde U\|\leq\varepsilon_U$ (where all subroutines are implemented exactly). Furthermore, let $V$ denote the unitary which will ultimately be executed by the quantum computer, i.e., the unitary which is obtained after all decomposition rules and approximations have been applied. By the triangle inequality and Lemma~\ref{lm:unitaryerror},
\begin{align*}
	\|U-V\|&\leq\|U-\tilde U\|+\|\tilde U-V\|\\&\leq \varepsilon_U + \sum_{S\in\mathcal S(U,\varepsilon_U)}\varepsilon_S f_S(\varepsilon_U)\\&\leq \varepsilon
\end{align*}
\end{proof}
In Fig.~\ref{fig:compilation}, for example, the left-most $^cU$ box gets $\varepsilon_1$ as its error budget. Depending on the implementation details of $^cU$, some of this budget may already be used to decompose $^cU$ into its subroutines, even assuming that all subroutines of $^cU$ are implemented exactly. The remaining error budget is then distributed among its subroutines, which is exactly the statement of the above theorem.

\begin{figure*}[t]
	\includegraphics[width=.9\linewidth]{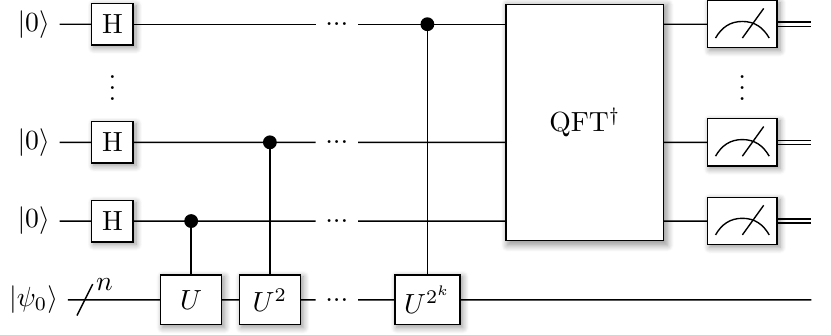}
	\caption{Quantum circuit of a quantum phase estimation applied to the time evolution operator $U=e^{-itH}$, where $H$ is the Hamiltonian of the quantum system being simulated, e.g., a transverse-field Ising model as in the text. After the inverse quantum Fourier transform (QFT$^\dagger$), a measurement yields the phase which was picked up by the input state. For the ground state $\ket{\psi_0}$, this is $U\ket{\psi_0}=e^{-iHt}\ket{\psi_0}=e^{-iE_0 t}\ket{\psi_0}$, allowing to extract (a $(k+1)$-bit approximation of) the energy $E_0$ of $\ket{\psi_0}$.}
	\label{fig:qpe}
\end{figure*}

The decomposition of the cost can be performed at different levels of granularity. This translates into, e.g., having a larger set $\mathcal S(U,\varepsilon)$ and more functions $f_S(\varepsilon_U)$ that are equal to 1. The two extreme cases are
\begin{enumerate}
	\item $f_S(\varepsilon)=1\;\forall S\in\mathcal S(U,\varepsilon)$, $|\mathcal S(U,\varepsilon)|=\#$gates needed to implement $U$:
	
	A different $\varepsilon_U$ for each gate
	
	\item $f_S(\varepsilon)=\#$gates needed to implement $U\;\forall S\in\mathcal S(U,\varepsilon)$, $|\mathcal S(U,\varepsilon)|=1$:
	
	The same $\varepsilon_\varnothing$ for all gates
	
\end{enumerate}
Therefore, this solves the first issue of Lemma~\ref{lm:unitaryerror}: In a practical implementation, the size of the set $\mathcal S(U,\varepsilon)$ can be adaptively chosen such that the resulting optimization problem which is of the form
\begin{align*}
&{ }(\varepsilon_{S_1}^\star,\cdots,\varepsilon_{S_N}^\star)\in\arg\min C_\text{Program}(\varepsilon_{S_1}^{},\cdots,\varepsilon_{S_N}^{})\\
&\text{such that }\; { }\varepsilon_\text{\tiny Program}(\varepsilon_{S_1}^{\star},\cdots,\varepsilon_{S_N}^{\star})\leq \varepsilon
\end{align*}
for a user- or application-defined over-all tolerance $\varepsilon$, can be solved using a reasonable amount of resources. Moreover, the costs of optimization can be reduced by initializing the initial trial parameters $\varepsilon_{S_i}^{}$ to the corresponding solution accuracies of a lower-dimensional optimization problem where $\mathcal S(U,\varepsilon)$ had fewer distinct subroutines. This approach is similar to multi-grid schemes which are used to solve partial differential equations.

The second issue with a direct application of Lemma~\ref{lm:unitaryerror} is the varying number of optimization parameters, which is also resolved by Theorem~\ref{thm:decomp}. Of course one can simply make $\mathcal S(U,\varepsilon)$ tremendously large such that most of the corresponding $f_S(\varepsilon)$ are zero. This, however, is a rather inefficient solution which would also be possible when using Lemma~\ref{lm:unitaryerror} directly. A better approach is to inspect $\mathcal S(U,\varepsilon)$ for different $\varepsilon$ and to then choose $A$ auxiliary subroutine sets $S^a_1,...,S^a_A$ such that each additional subroutine $V_k^a$ which appears when changing $\varepsilon$ (but is not a member of any $S$ of the original $\mathcal S(U,\varepsilon)$) falls into exactly one of these sets. The original set $\mathcal S(U,\varepsilon)$ can then be extended by these auxiliary sets before running the optimization procedure. Again, the level of granularity of these auxiliary sets and thus the number of such sets $A$ can be tuned according to the resources that are available to solve the resulting optimization problem.

\section{Example application}\label{sec:tfim}

As an example application, we consider the simulation of a quantum mechanical system called the \textit{transverse-field Ising model} (TFIM), which is governed by the Hamiltonian
\[
	\hat H=-\sum_{\langle i,j\rangle} J_{ij}\sigma_z^{i}\sigma_z^{j}-\sum_i\Gamma_i\sigma_x^i,
\]
where $J_{ij}$ are coupling constants and $\Gamma_i$ denotes the strength of the transverse field at location $i$. $\sigma_x^i$ and $\sigma_z^i$ are the Pauli matrices, i.e.,
\[
	\sigma_x=\left(\begin{matrix}0 & 1\\1 & 0\end{matrix}\right)\;\;\text{ and }\;\;
	\sigma_z=\left(\begin{matrix}1 & 0\\0 & -1\end{matrix}\right)
\]
acting on the $i$-th spin. The sum over $\langle i,j\rangle$ loops over all pairs of sites $(i,j)$ which are connected. In our example, this corresponds to nearest-neighbor sites on a one-dimensional spin chain (with periodic boundary conditions) of length $N$. Given an approximation $\ket{\tilde{\psi_0}}$ to the ground state $\ket{\psi_0}$ of $\hat H$, we would like to determine the ground state energy $E_0$ such that
\[
	\hat H\ket{\psi_0}=E_0\ket{\psi_0}.
\]
It is well-known that quantum phase estimation (QPE) can be used to achieve this task which leads to a general circuit structure as in Fig.~\ref{fig:qpe}. 

\vspace{10pt}\noindent\textbf{Individual compilation stages.}
We now analyze the QPE algorithm for TFIM ground state estimation and the resulting optimization problem for approximation errors. First note that if the overlap between $\ket{\psi_0}$ and $\ket{\tilde\psi_0}$ is large, a successful application of QPE followed by a measurement of the energy register will collapse the state vector onto $\ket{\psi_0}$ and output $E_0$ with high probability (namely $p=|\braket{\tilde\psi_0|\psi_0}|^2$). 

There are various ways to implement QPE~\cite{nielsen2002quantum}, but the simplest to analyze is the coherent QPE followed by a measurement of all control qubits, see Fig.~\ref{fig:qpe} for an illustration of the circuit. This procedure requires ${16\pi}/{\varepsilon_\text{QPE}}$ applications of (the controlled version of) the time-evolution operator $U_\delta=\exp(-i\delta\hat H)$ for a success probability of $1/2$, where $\varepsilon_\text{QPE}$ denotes the desired accuracy (bit-resolution of the resulting eigenvalues)~\cite{Reiher7555}. Using a Trotter decomposition of $U_\delta$, i.e., for large $M$
\begin{align*}
	U_\delta&\approx \left( U^J_\frac \delta M U^\Gamma_\frac\delta M\right)^M\\&=\left( e^{-i\frac\delta M\sum_{i}J_{i,i+1}\sigma_z^i\sigma_z^{i+1}} e^{-i\frac\delta M\sum_i \Gamma_i\sigma_x^i}\right)^M\\
	&=\left( \prod_ie^{-i\frac\delta MJ_{i,i+1}\sigma_z^i\sigma_z^{i+1}} \prod_ie^{-i\frac\delta M\Gamma_i\sigma_x^i}\right)^M,
\end{align*}
allows to implement the global propagator $U_\delta$ using a sequence of local operations. These consist of z- and x-rotations in addition to nearest-neighbor CNOT gates to compute the parity (before the z-rotation and again after the z-rotation to uncompute the parity). The rotation angles are $\theta_z=2\frac{\delta}{M}J_{i,i+1}$ and $\theta_x=-2\frac{\delta}{M}\Gamma_i$ for z- and x-rotations, respectively. The extra factor of two arises from the the definitions of the Rz and Rx gates, see Sec.~\ref{sec:origin}.

In order to apply error correction to run the resulting circuit on actual hardware, these rotations can be decomposed into a sequence of Clifford+$T$ gates using rotation synthesis. Such a discrete approximation up to an accuracy of $\varepsilon_R$ features $\mathcal O(\log\varepsilon_R^{-1})$ T-gates if the algorithms in~\cite{ross2014optimal,KMM1} are used, where even the constants hidden in the $\mathcal O$ notation were explicitly determined.

\vspace{10pt}\noindent\textbf{Casting the example into our framework.}
The first compilation step is to resolve the QPE library call. In this case, it is known that the cost of QPE applied to a general propagator $U$ is
\[
	C(\text{QPE}_U, \varepsilon) = \frac{16\pi}{\varepsilon_\text{QPE}} C(^cU,\varepsilon_U),
\]
where $^cU$ denotes the controlled version of the unitary $U$, i.e.,
\[
	^cU := \Ket 0\Bra 0\otimes\mathbbm 1 + \Ket 1\Bra 1\otimes U.
\]
Furthermore, the chosen tolerances must satisfy
\[
	\frac{16\pi}{\varepsilon_\text{QPE}}\varepsilon_U \leq \varepsilon - \varepsilon_\text{QPE}.
\]
The next step is to approximate the propagator using a Trotter decomposition. Depending on the order of the Trotter formula being used, this yields
\begin{align*}
	C(^cU,\varepsilon_U) &= M(\varepsilon_\text{Trotter}) (C(^cU_1,\varepsilon_{U_1}) + C(^cU_2,\varepsilon_{U_2}))\\&\text{ with }\; M(\varepsilon_\text{Trotter})(\varepsilon_{U_1}+\varepsilon_{U_2})\leq\varepsilon_U-\varepsilon_\text{Trotter}.
\end{align*}
In the experiments section, we will choose $M(\varepsilon_{\text{Trotter}})\propto\frac 1{\sqrt{\varepsilon_{\text{Trotter}}}}$ as an example. Finally, approximating the (controlled) rotations in $^cU_1$ and $^cU_2$ by employing rotation synthesis,
\begin{align*}
	C(^cU_i,\varepsilon_{U_i}) = 2N\cdot 4\log\varepsilon_R^{-1}\\\text{ with }\; 2N\varepsilon_R\leq\varepsilon_{U_i}\;\text{ for }i\in\{1,2\}.
\end{align*}
Collecting all of these terms and using that $C(^cU_1,\cdot) = C(^cU_2,\cdot)$ yields
\begin{align*}
C(\text{QPE}_U,\varepsilon) = \frac{16\pi}{\varepsilon_\text{QPE}}M(\varepsilon_\text{Trotter})\cdot 2\cdot 2N\cdot 4\log\varepsilon_R^{-1},\\
\text{ with }\varepsilon_\text{QPE} + \frac{16\pi}{\varepsilon_\text{QPE}}(2M(\varepsilon_\text{Trotter})\cdot 2N\varepsilon_R + \varepsilon_\text{Trotter})\leq\varepsilon.
\end{align*}
Note that this example is a typical application for a quantum computer which at the same time can serve as a proxy for other, more complex simulation algorithms. 

While the individual compilation stages may be different for other applications, the basic principle of iterative decomposition and approximation during compilation is ubiquitous. In particular, a similar compilation procedure would be employed when performing quantum chemistry simulations, be it using a Trotter-based approach~\cite{Wecker14a} or an approach that is based on a truncated Taylor series~\cite{berry2015simulating}.

\lstset{language=python,numbers=left,numberstyle=\footnotesize,frame=single,tabsize=4}
\begin{lstlisting}[float=t,escapeinside={(*}{*)},label=lst:annealingcode,caption={High-level description of the annealing-based algorithm to solve the resulting optimization problem. The actual implementation features different scaling constants for $\Delta E$ depending on the mode (error reduction vs. cost reduction).},morekeywords={rnd,floor},captionpos=b,numberstyle=\footnotesize\textcolor{gray},numbers=none,frame=top bottom,commentstyle=\itshape\textcolor{gray},basicstyle=\footnotesize]
(*$\beta$*) = 0
(*$\overline\varepsilon$*) = (*$\left(0.1, 0.1, \cdots, 0.1\right)$*)
cost = get_cost((*$\overline\varepsilon$*))
error = get_total_error((*$\overline\varepsilon$*))
for step in range(num_steps):
	i = floor(rnd() * len(eps))
	old_(*$\overline\varepsilon$*) = (*$\overline\varepsilon$*)
	if rnd() < 0.5:
		(*$\overline\varepsilon_i$*) *= 1 + (1 - rnd()) * (*$\delta$*)
	else:
		(*$\overline\varepsilon_i$*) /= 1 + (1 - rnd()) * (*$\delta$*)
	
	if error <= goal_error:
		# reduce cost
		(*$\Delta E$*) = get_cost((*$\overline\varepsilon$*)) - cost
	else:
		# reduce error
		(*$\Delta E$*) = get_total_error((*$\overline\varepsilon$*)) - error
	(*$p_\text{accept}$*) = min(1, (*$e^{-\beta \Delta E}$*))
	if rnd() > p(*$_\text{accept}$*):
		(*$\overline\varepsilon$*) = old_(*$\overline\varepsilon$*)
	(*$\beta$*) += (*$\Delta\beta$*)
\end{lstlisting}

\bigskip
\section{Implementation and numerical results}\label{sec:results}

In this section, we present implementation details and numerical results of our error management module. While the optimization procedure becomes harder for fine-grained cost and error analysis, the benefits in terms of the cost of the resulting circuit are substantial.

\vspace{10pt}\noindent\textbf{Optimization methodology.}
We use a two-mode annealing procedure for optimization, in which two objective functions are reduced as follows: The first mode is active whenever the current overall error is larger than the target accuracy $\varepsilon$. In this case, it performs annealing until the target accuracy has been reached. At this point, the second mode becomes active. It performs annealing-based optimization to reduce the circuit cost function. After each such step, it switches back to the error-reduction subroutine if the overall error increased above $\varepsilon$.

Both annealing-based optimization modes follow the same scheme, which consists of increasing/decreasing a randomly chosen $\varepsilon_i$ by multiplying/dividing it by a random factor $f\in(1,1+\delta]$, where $\delta$ can be tuned to achieve an acceptance rate of roughly $50\%$. Then, the new objective function value is determined, followed by either a rejection of the proposed change in $\varepsilon_i$ or an acceptance with probability
\[
	p_{\text{accept}}=\min(1, e^{-\beta\Delta E}),
\]
where $\beta=T^{-1}$ and $T$ denotes the annealing temperature. This means, in particular, that moves which do not increase the energy, i.e., $\Delta E \leq 0$ are always accepted. The pseudo-code of this algorithm can be found in Listing~\ref{lst:annealingcode}.

\begin{figure}[htb]
\centering
\subfigure[{Additional optimization allows to reduce the circuit cost by almost a factor of two over the first encountered feasible solution (see inset).\label{fig:erroronly}}]{
\resizebox{\linewidth}{!}{\input{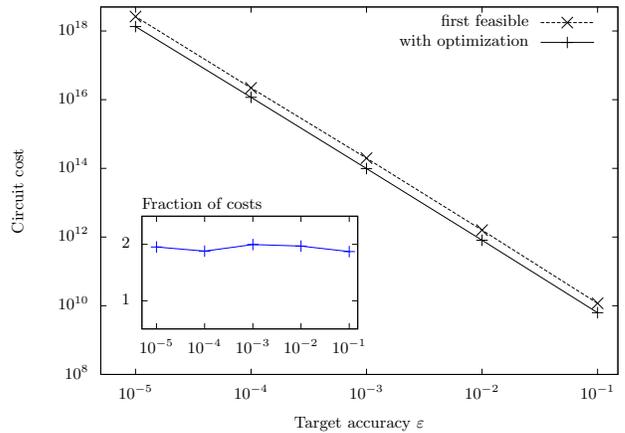}}}
\subfigure[{Performing a three-variable optimization enables a reduction of the resulting circuit cost by several orders of magnitude when compared to two-variable optimization.\label{fig:numparams}}]{
\resizebox{\linewidth}{!}{\input{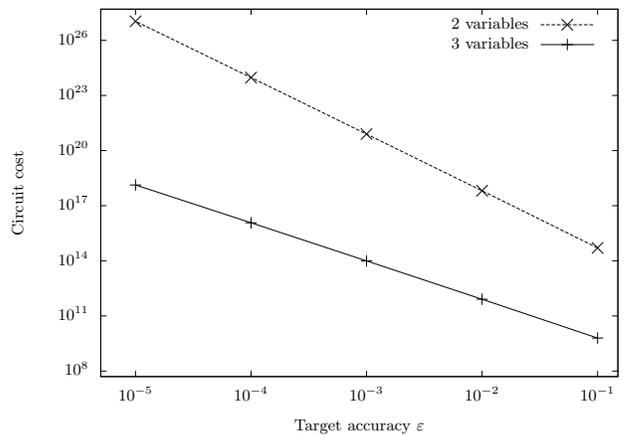}}}

\caption{Numerical results for the optimization problem resulting from the transverse-field Ising model example discussed in Sec.~\ref{sec:tfim}. Improving the first encountered feasible solution by further optimization allows to reduce the cost by almost a factor of two and the number of different parameters can influence the resulting cost by several orders of magnitude.}
\end{figure}

\vspace{10pt}\noindent\textbf{Results.} Using the example of a transverse-field Ising model which was discussed in Sec.~\ref{sec:tfim}, we determine the benefits of our error management module by running two experiments. The first experiment aims at assessing the difference between a feasible solution, i.e., values $\varepsilon_i$ which produce an overall error that is less than the user-defined tolerance, and an optimized feasible solution. In the first case, we only run the first mode until a feasible solution is obtained and in the latter, we employ both modes as outlined above. Fig.~\ref{fig:erroronly} depicts the costs of the resulting circuit as a function of the desired overall accuracy $\varepsilon$. 

\begin{figure}[tb]
\centering
\resizebox{\linewidth}{!}{\input{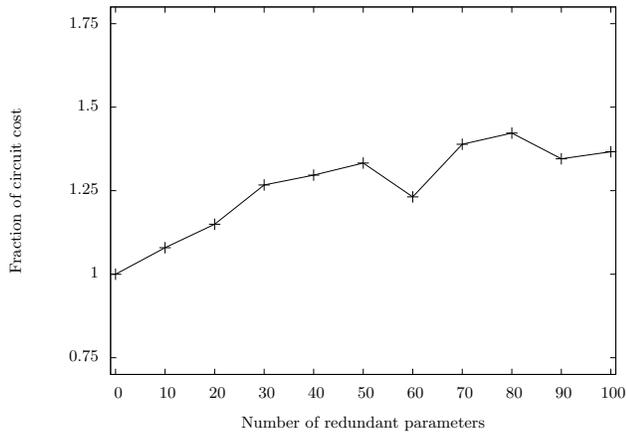}}
\caption{The circuit cost with $i\in\{0,10,...,100\}$ redundant parameters divided by the cost achieved with no redundancies. As expected, the problem becomes harder to optimize as more parameters are added. The best result of 1000 runs (with different random number seeds) is reported. 5000 annealing steps to $\beta_{max}=10$ were performed for each run.}
\label{fig:rob1}
\end{figure}

\begin{figure}[b]
\centering
\resizebox{\linewidth}{!}{\input{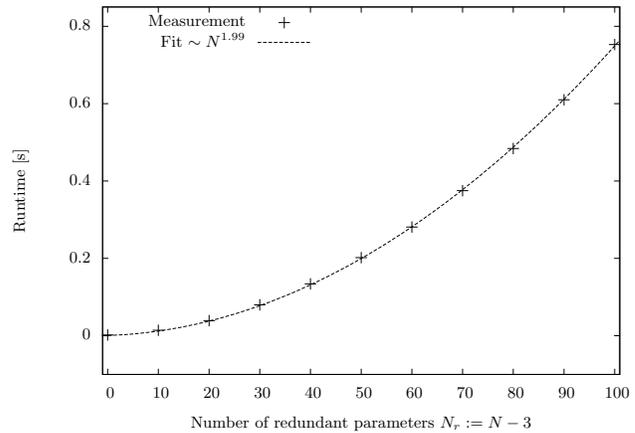}}
\caption{Runtime for finding an initial feasible solution which was then optimized further in order to reduce the circuit cost. As expected, the time increases quadratically with the number of parameters.\vspace{10pt}}
\label{fig:rob2}
\end{figure}

The second experiment aims to show the benefit of an increased number of $\varepsilon_i$ parameters in the same example. The difference between the circuit costs when using just two such parameters (i.e., setting $\varepsilon_R=\varepsilon_{\text{Trotter}}$) versus using all three is depicted in Fig.~\ref{fig:numparams}.

Finally, we measure the robustness of the optimization procedure by introducing redundant parameters, i.e., additional rotation gate synthesis tolerances $\varepsilon_{R_i}$, where the optimal choice would be $\varepsilon_R=\varepsilon_{R_i}=\varepsilon_{R_j}$ for all $i,j$. However, because the resulting optimization problem features more parameters, it is harder to solve and the final circuit cost is expected to be higher.
In addition, the time it takes to find an initial feasible solution will grow. See Figs.~\ref{fig:rob1} and~\ref{fig:rob2} for the results which indicate that this approach is scalable to hundreds of variables if the goal is to find a feasible solution. However, as the number of parameters grows, it becomes increasingly harder to simultaneously optimize for the cost of the circuit. This could be observed, e.g., with 100 additional (redundant) parameters, where further optimization of the feasible solution reduced the cost from $1.65908\cdot 10^{12}$ to $1.10752\cdot 10^{12}$, which is far from the almost $2x$ improvement which was observed for smaller systems in Fig.~\ref{fig:erroronly}. Also, the scaling of the runtime in Fig.~\ref{fig:rob2} can be explained since new updates are proposed by selecting $i\in [0,...,N-1]$ uniformly at random (followed by either increasing or decreasing $\varepsilon_i$). Due to this random walk over $i\in[0,...,N-1]$, the overall runtime is also expected to behave like the expected runtime of a random walk and, therefore, to be in $\mathcal O(N^2)$.

\section{Summary and future work}\label{sec:outlook}
We have presented a methodology for managing approximation errors in compiling quantum algorithms. Given that the way in which the overall target error is distributed among subroutines greatly influences the resource requirements, it is crucial to optimize this process, in particular for large-scale quantum algorithms that are composed of many subroutines. We leverage an annealing-based procedure to find an initial feasible solution which then is optimized further. 

Our scheme for error management only addresses errors that occur in approximations during the compilation process into a fault-tolerant gate set. Future work might include hardware errors, e.g., systematic over- or under-rotations of gates performed by the target device. Furthermore, additional numerical studies for various quantum algorithms can be performed in order to arrive at heuristics for choosing the number of optimization parameters. Moreover, building on our error management methodology, one can automate the entire process of resource estimation for certain subclasses of quantum algorithms. This would yield a useful tool for assessing the practicality of known quantum algorithms, similar to the analysis carried out manually in~\cite{Reiher7555}.

\end{document}